\documentclass[10pt,aps,twocolumn,longbibliography,prx]{revtex4-2}
\usepackage[english]{babel}
\usepackage[utf8]{inputenc}
\usepackage[colorinlistoftodos, color=green!40, prependcaption]{todonotes}
\usepackage{amsthm}
\usepackage{comment}
\usepackage{mathtools}
\usepackage{physics}
\usepackage{xcolor}
\usepackage{graphicx}
\usepackage[export]{adjustbox}
\usepackage{graphbox}
\usepackage{placeins}
\usepackage[T1]{fontenc}
\usepackage{lipsum}
\usepackage{csquotes}
\usepackage[pdftex, pdftitle={Article}, pdfauthor={Author}]{hyperref} 
\usepackage{amssymb}
\usepackage{amsfonts}
\usepackage{amsmath}
\usepackage{stmaryrd}
\usepackage{braket}
\usepackage{dsfont}
\usepackage{bbold}
\usepackage{relsize}
\usepackage{float}
\usepackage[font=scriptsize]{caption}
\let\ACMmaketitle=\maketitle
\renewcommand{\maketitle}{\begingroup\let\footnote=\thanks \ACMmaketitle\endgroup}

\newcommand{\overbar}[1]{\mkern 1.5mu\overline{\mkern-1.5mu#1\mkern-1.5mu}\mkern 1.5mu}
\newcommand{\M}[1]{\mathcal{M}_{#1}}
\newcommand{\C}[1]{\mathbb{C}^{#1}}

\newtheorem{theorem}{Theorem}[section]
\newtheorem{definition}[theorem]{Definition}
\newtheorem*{definition*}{Definition}
\newtheorem{proposition}[theorem]{Proposition}
\newtheorem{corollary}[theorem]{Corollary}

\newtheorem{remark}[theorem]{Remark}

\newtheorem*{conjecture*}{Conjecture}

\theoremstyle{definition}
\newtheorem{example}[theorem]{Example}

\newcommand{\PRLsep}{\noindent\makebox[\linewidth]{\resizebox{0.3333\linewidth}{1pt}{$\blacklozenge$}}\bigskip}

\begin{document}
%TC:ignore
    \title{Can entanglement hide behind triangle-free graphs?}
    \thanks{Spoiler alert: Yes, it can!}
    \author{Satvik Singh}
    \email[Correspondence email address: ]{satviksingh2@gmail.com}
    \affiliation{Department of Physical Sciences, Indian Institute of Science Education and Research (IISER) Mohali, Punjab, India}

\date{\today}

\begin{abstract}

We present a novel approach to unveil a new kind of entanglement in bipartite quantum states whose diagonal zero-patterns in suitable matrix representations admit a nice description in terms of triangle-free graphs. Upon application of a local averaging operation, the separability of such states transforms into a simple matrix positivity condition, violation of which implies the presence of entanglement. We completely characterize the class of triangle-free graphs which allows for non-trivial entanglement detection using the above test. Moreover, we develop a recipe to construct a plethora of unique classes of PPT entangled triangle-free states in arbitrary dimensions. Finally, we link the task of entanglement detection in general states to the well-known graph-theoretic problem of finding triangle-free induced subgraphs in a given graph.

\end{abstract}

\maketitle
%TC:endignore
\section{Introduction}
Termed ``spooky action at a distance'' by Einstein, quantum entanglement has come a long way since its inception in the first half of the 20\textsuperscript{th} century \cite{Einstein1935spooky,schrodinger1935entanglement} --- both in terms of theoretical and experimental relevance. Now understood as one of the most fundamental non-classical features of quantum theory, entanglement has garnered a reputation as a resource of immense practical worth, with its groundbreaking effects visible in a wide array of quantum cryptographic \cite{Ekert1991crypto, Vazirani2014crypto}, teleportation \cite{Bennet1993teleportation} and computation \cite{Raussendorf2003comp} protocols. However, even after more than 85 years of its discovery, the theory of entanglement has managed to retain its richness and complexity.

First concrete signs of entanglement were explored by Bell in terms of violations of spin correlation inequalities \cite{Bell1964inequalities}, which were later shown by Werner to be sufficient but not necessary to detect entanglement \cite{Werner1989ent}. Connections with the theory of positive maps were not made until the late `90s, when positivity under partial transposition (PPT) was shown by Peres to be a necessary condition for separability \cite{Peres1996PPT}, which was later also proved to be sufficient in low dimensional $2\otimes2$ and $2\otimes 3$ systems \cite{Horodecki1996PPT}. Since then, the study of separability has witnessed tremendous growth, with a myriad of tests now available to detect entanglement in various different scenarios, see \cite[Section VI B]{Horodecki2009review} and \cite{Chruscinski2014positive}. The primary goal of all these tests is to bypass a major theoretical hurdle which is intrinsic to the structure of entanglement itself, namely, the NP-hardness of the weak membership problem for the convex set of separable bipartite states \cite{Gurvits2003NPhard,Sevag2010NPhard}. Put simply, it is not possible for any classical algorithm to efficiently determine whether a given state is entangled or not (assuming that the widely believed P$\neq$NP result holds).
Still, several algorithms with complexities scaling exponentially with the system dimension do exist \cite{Lewenstein1998alg,Doherty2004alg, Eisert2004alg}. 

In this work, we unearth a new kind of entanglement in bipartite states, which is easy to detect but rich enough to manifest itself in astonishing variety. For a $d\otimes d$ state, this entanglement camouflages beneath an intriguing design of zeros on the state's diagonal, which, when rearranged in the form of a $d\times d$ matrix, reveals a distinctive ``triangle-free'' ($\Delta$-free) zero-pattern in the off-diagonal part of the matrix. The terminology is akin to the one used in graph theory, where $\Delta$-free graphs have been a subject of interest for well over a century. To detect entanglement in these $\Delta$-free states, we project them onto the subspace of local diagonal orthogonal invariant (LDOI) matrices, through a local operation which preserves the separability and $\Delta$-free property of states. LDOI states have been amply scrutinized in literature, \cite{Chruscinski2006PPTstates, Johnston2019pairwise, Nechita2020graphical, Singh2020diagonal}, since many important examples of PPT (and NPT) states are of this type: Werner and Isotropic states \cite{Werner1989ent, Horodecki1999iso}, Diagonal Symmetric states \cite{Yu2016Dicke,tura2018Dicke}, canonical NPT states \cite{Shor2000NPT}, to name a few (see also \cite[Section 3]{Singh2020diagonal}). Separable states in this class admit an equivalent description in terms of the cone of triplewise completely positive matrices, which is a generalization of the well-studied cone of completely positive matrices \cite{abraham2003completely}. The already established significance of $\Delta$-free graphs within the theory of completely positive matrices \cite[Section 2.4]{abraham2003completely} is the primary source of inspiration for the results we present here. We will see that separability in $\Delta$-free LDOI states materializes into simple positivity conditions on certain associated matrices, enabling one to easily detect entanglement in such states. As the number of $\Delta$-free zero patterns increases rather tremendously with the system's dimensions, so does the number of distinct $\Delta$-free entangled families of states. For perspective, in a $15\otimes 15$ system, we'll provide an explicit way to construct $\sim 10^{10}$ distinct families of PPT entangled $\Delta$-free states. Because of this sheer diversity, the traditional methods for entanglement detection get crippled in the regime of $\Delta$-free states. In contrast, the simplicity of our method cannot be overstated, which provides a highly non-trivial yet computationally efficient technique for entanglement detection in a wide array of scenarios. \\
\phantom{--}Let us now briefly comment on this paper's organization. In Section~\ref{sec:preliminaries}, we review the theory of LDOI states and graphs. Section~\ref{sec:main} contains the primary entanglement test and hence forms the core of our work. In this section, we also classify the subset of $\Delta$-free graphs which allows for the possibility of non-trivial entanglement detection using our test. A systematic scheme to construct new families of PPT entangled $\Delta$-free states in arbitrary dimensions is articulated in Section \ref{sec:PPTent}. Section~\ref{sec:beyond} discusses entanglement detection in non $\Delta$-free states, and connects it to the triangle-free induced subgraph problem. Finally, we present a summary of our results and important directions for future work in Section~\ref{sec:conclusion}

\section{Preliminaries} \label{sec:preliminaries}
\subsection{Local diagonal orthogonal invariant states}
We will exclusively be dealing with the finite $d$-dimensional complex Hilbert space $\C{d}$ and the space of $d\times d$ complex matrices $\M{d}$, with the canonical bases $\{\ket{i}\}_{i=1}^d$ and $\{\ketbra{i}{j} \}_{i,j=1}^d$, respectively. \emph{Positive semi-definite} and \emph{entrywise non-negative} matrices in $\M{d}$ will be denoted by $A\geq 0$ and $A\succcurlyeq 0$ respectively. \emph{Quantum states} ($A\geq 0$, $\operatorname{Tr}A=1$) will be denoted by $\rho$. We define $[d]\coloneqq \{1,2,\ldots ,d\}$. The following class of bipartite states in $\M{d}\otimes \M{d}$ will play a crucial role in this paper:
\begin{align} \label{eq:LDOI}
    \rho_{A,B,C} = \sum_{i,j=1}^d A_{ij} \ketbra{ij}{ij} &+ \sum_{1\leq i\neq j\leq d} B_{ij} \ketbra{ii}{jj} \nonumber \\
    &+ \sum_{1\leq i\neq j\leq d} C_{ij} \ketbra{ij}{ji}
\end{align}
where $A,B,C\in \M{d}$ are matrices such that $\operatorname{diag}A=\operatorname{diag}B=\operatorname{diag}C$, $A\succcurlyeq 0$, $B\geq 0$, $C=C^{\dagger}$, $A_{ij}A_{ji}\geq |C_{ij}|^2 \,\, \forall i,j\in [d]$ and $\sum_{ij}A_{ij}=1$. These conditions ensure that $\rho_{A,B,C}$ is indeed a quantum state, as can easily be checked. If the partial transpose $\rho_{A,B,C}^\Gamma\geq 0$ as well (i.e. $\rho$ is PPT), then $C\geq 0$ and $A_{ij}A_{ji}\geq |B_{ij}|^2 \,\, \forall i,j\in [d]$ \cite[Lemma 2.12, 2.13]{Singh2020diagonal}. These states enjoy a special \emph{local diagonal orthogonal invariance} (LDOI) property:
\begin{equation} \label{eq:LDOI-prop}
 \forall O\in \mathcal{DO}_d, \quad  \rho_{A,B,C} = (O\otimes O)\rho_{A,B,C}(O\otimes O)  
\end{equation}
where the group of \emph{diagonal orthogonal} matrices in $\M{d}$ is denoted by $\mathcal{DO}_d$. For an arbitrary state $\rho\in \M{d}\otimes \M{d}$, we define matrices $A,B,C\in \M{d}$ entrywise as $A_{ij}=\langle ij|\rho|ij\rangle$, $B_{ij}=\langle ii|\rho|jj\rangle$, and $C_{ij}=\langle ij|\rho|ji\rangle$ for $i,j\in [d]$. Then, the local averaging operation in Eq.~\eqref{eq:LDOI-proj} acts as the orthogonal projection onto the subspace of LDOI matrices, where $O$ is a random diagonal orthogonal matrix with uniformly random signs $\{ \pm \}$ on its diagonal, which are independent and identically distributed (i.i.d):
\begin{equation} \label{eq:LDOI-proj}
    \rho \mapsto \mathbb{E}_O [(O\otimes O) \rho (O\otimes O)] = \rho_{A,B,C}
\end{equation}
The validity of Eq.~\eqref{eq:LDOI-proj} can be established by first showing that the LDOI property of a state (as stated in Eq.~\eqref{eq:LDOI-prop}) is equivalent to its invariance under the aforementioned local averaging operation and then by checking that all matrix units in $\M{d}\otimes \M{d}$ vanish under this operation except the ones of the form $\ketbra{ij}{ij}, \ketbra{ii}{jj}$ or $\ketbra{ij}{ji}$, which stay invariant. The interested reader should refer to \cite[Sections 6,7]{Nechita2020graphical} to gather more details about this argument.

Now, if $C$ is diagonal in Eq.~\eqref{eq:LDOI}, we obtain a subclass of LDOI states, which is defined by the \emph{conjugate local diagonal unitary invariance} (CLDUI) property:
\begin{align}
    \rho_{A,B} &= \sum_{i,j=1}^d A_{ij} \ketbra{ij}{ij} + \sum_{1\leq i\neq j\leq d} B_{ij} \ketbra{ii}{jj} \label{eq:CLDUI}
\end{align}
\begin{equation}
    \forall U\in \mathcal{DU}_d, \quad \rho_{A,B} = (U\otimes U^{\dagger})\rho_{A,B}(U^{\dagger}\otimes U)  
\end{equation}
where $\mathcal{DU}_d$ is the group of \emph{diagonal unitary} matrices in $\M{d}$. Analogous to Eq.~\eqref{eq:LDOI-proj}, the local averaging operation in Eq.~\eqref{eq:CLDUI-proj} defines the orthogonal projection onto the subspace of CLDUI matrices, where $U$ is a random diagonal unitary matrix having uniform, i.i.d entries on the unit circle in $\C{}$ and $A,B\in \M{d}$ are as defined before.
\begin{equation} \label{eq:CLDUI-proj}
    \rho \mapsto \mathbb{E}_U [(U\otimes U^{\dagger}) \rho (U^{\dagger}\otimes U)] = \rho_{A,B}
\end{equation}
\begin{remark}
If $B$ is diagonal in Eq.~\eqref{eq:LDOI}, we get the subclass of \emph{local diagonal unitary invariant (LDUI)} states, which we choose not to deal with in this article. Since the two classes are linked through the operation of partial transposition, the separability results for \emph{CLDUI} states will identically apply to \emph{LDUI} states as well.
\end{remark}

We now define the notions of \emph{pairwise} and \emph{triplewise completely positive} (PCP and TCP) matrices, which are fundamentally connected to the separability of the CLDUI and LDOI states, respectively. For $\ket{v},\ket{w}\in \C{d}$, we denote the operations of entrywise complex conjugate and Hadamard product in $\C{d}$ by $\ket{\overbar{v}}$ and $\ket{v\odot w}$, respectively. Recall that $\rho\in \M{d}\otimes \M{d}$ (un-normalized) is said to be \emph{separable} if there exist a finite set of vectors $\{\ket{v_k},\ket{w_k}\}_{k\in I}\subset \C{d}$ such that $\rho = \sum_{k\in I}\ketbra{v_k w_k}{v_k w_k}$. 
\begin{definition}\label{def:PCP-TCP}
Let $A,B,C\in \M{d}$. Then, \\[0.2cm] 
$\bullet\,\, (A,B)$ is said to be \emph{pairwise completely positive (PCP)} if there exist a finite set of vectors $\{\ket{v_k},\ket{w_k}\}_{k\in I}\subset \C{d}$ such that
\begin{equation*}
    A = \sum_{k\in I} |v_k \odot \overbar{v_k}\rangle\langle w_k \odot \overbar{w_k}| \quad B = \sum_{k\in I} \ketbra{v_k \odot w_k}{v_k \odot w_k} 
\end{equation*}
$\bullet\,\, (A,B,C)$ is said to be \emph{triplewise completely positive (TCP)} if there exist a finite set of vectors $\{\ket{v_k},\ket{w_k}\}_{k\in I}\subset \C{d}$ such that
\begin{align*}
      A = \sum_{k\in I} |v_k \odot \overbar{v_k}\rangle&\langle w_k \odot \overbar{w_k}| \quad B = \sum_{k\in I} \ketbra{v_k \odot w_k}{v_k \odot w_k}   \\
      &C = \sum_{k\in I} \ketbra{v_k \odot \overbar{w_k}}{v_k \odot \overbar{w_k}} 
\end{align*}
\end{definition}

\begin{theorem}
For $A,B,C\in \M{d}$ with equal diagonals,
\begin{itemize}
    \item $\rho_{A,B}$ is separable $\iff (A,B)$ is \emph{PCP}.
    \item $\rho_{A,B,C}$ is separable $\iff (A,B,C)$ is \emph{TCP}.
\end{itemize}
\end{theorem}
\begin{proof}
We will prove the theorem for an unnormalized LDOI state $\rho_{A,B,C}$. Assume first that the state is separable and hence admits a decomposition of the form $\rho_{A,B,C}=\sum_{k\in I}\ketbra{v_k w_k}{v_k w_k}$. Then, from the definition of the associated $A,B,C$ matrices, it is evident that they admit the desired TCP decomposition with vectors $\{\ket{v_k},\ket{w_k} \}_{k\in I}$. Conversely, assume that $(A,B,C)$ is TCP with the vectors $\{\ket{v_k},\ket{w_k} \}_{k\in I}$ forming its TCP decomposition as in Definition~\ref{def:PCP-TCP}. Now, construct the separable state $\rho=\sum_{k\in I}\ketbra{v_k w_k}{v_k w_k}$. It is then easy to see that $\rho_{A,B,C}$ is separable, since it can be obtained from $\rho$ by a local projection: $\rho_{A,B,C} = \mathbb{E}_O [(O\otimes O) \rho (O\otimes O)]$, see Eq.~\eqref{eq:LDOI-proj}.
\end{proof}

For a more thorough analysis of the space of LDUI, CLDUI, and LDOI matrices, the readers should refer to \cite[Sections 2-5]{Singh2020diagonal}. In particular, \cite[Section 3]{Singh2020diagonal} contains a comprehensive list of many important examples of states (like the Werner and Isotropic states, mixtures of Dicke states, $3\otimes 3$ edges states, etc.) from the literature which lie in the LDOI class.

\subsection{Graphs}
A (simple, undirected) \emph{graph} $G$ consists of a finite \emph{vertex} set $V$, together with a finite set of two-element (unordered) subsets of $V$, known as the \emph{edge} set $E$. $G$ can be represented pictorially by drawing the elements (vertices) of $V$ as points, with $i,j\in V$ connected by a line if $\{i,j\}\in E$. For two graphs $G=(V_G,E_G)$ and $H=(V_H,E_H)$, we say that $G$ \emph{contains} $H$ ($H\subseteq G$) if $V_H\subseteq V_G$ and $E_H\subseteq E_G$. For $l\geq 3$, an $l$-cycle (denoted $C_l$) in a graph is a sequence of edges $\{i_0,i_1\}, \{i_1,i_2\},\ldots ,\{i_{l-1},i_l\}$ such that $i_1\neq i_2 \neq \ldots \neq i_l$, and $i_0=i_l$. $3$-cycles are called triangles. A graph which does not contain any triangles is called $\Delta$-\emph{free}. A graph is termed \emph{cyclic} if it contains a cycle and \emph{acyclic} otherwise. Given a matrix $A\in \M{d}$, we associate a graph $G(A)$ to it on $d$ vertices such that for $i\neq j$, $\{i,j\}$ is an edge if both $A_{ij}$ and $A_{ji}$ are non-zero. The \emph{adjacency matrix} $A\in \M{d}$ of a graph $G$ (denoted $\operatorname{ad}G$) on $d$ vertices is defined as follows: $\operatorname{diag}A = 0$ and $A_{ij}=A_{ji}=1$ if $\{i,j \}$ is an edge for $i\neq j$. If $G$ is a graph on $d$ vertices, then $B\in \M{d}$ is called a matrix realization of $G$ if $G(B)=G$. We refer the readers to the excellent book by Bondy and Murty \cite{Bondy2011Graph} for a more thorough introduction to graphs. 

\section{Main Results} \label{sec:main}
Maintaining close proximity with the graph-theoretic terminology, we begin by introducing the concept of \emph{triangle-free} ($\Delta$-\emph{free}) bipartite states.
\begin{definition} \label{def:triangle-free}
A bipartite state $\rho\in \M{d}\otimes \M{d}$ is said to be \emph{triangle-free} ($\Delta$-free) if the associated matrix $A\in \M{d}$ defined entrywise as $A_{ij}=\langle ij|\rho |ij\rangle$ for $i,j\in [d]$, is such that $G(A)$ is $\Delta$-free.
\end{definition}

From the definition, it is clear that the property of $\Delta$-freeness of a state is nothing but a statement on the zero pattern of the state's diagonal. It is not too difficult to ascertain whether a state has this property, as efficient polynomial-time algorithms exist to determine whether a graph $G$ is $\Delta$-free \cite{Alon1997trianglefree}; for example, $\operatorname{trace}[(\operatorname{ad}G)^3]=0\iff G$ is $\Delta$-free. However, it should be noted that the notion of $\Delta$-freeness is basis dependent, i.e. given a $\Delta$-free state $\rho\in \M{d}\otimes \M{d}$, there may exist unitary matrices $U,V$ such that $(U\otimes V)\rho (U\otimes V)^\dagger$ is \emph{not} $\Delta$-free.

We now proceed towards exploiting the $\Delta$-free property of an LDOI state to obtain a powerful necessary condition for its separability. Recall that for a vector $\ket{v}\in \C{d}$, $\operatorname{supp}\ket{v}\coloneqq \{i\in [d] : \langle i|v\rangle\neq 0\}$. We define $\sigma(v)$ to be the size of $\operatorname{supp}\ket{v}$. For $B\in \M{d}$, the comparison matrix $M(B)$ is defined as $M(B)_{ij}=|B_{ij}|$ for $i=j$ and $M(B)_{ij} = -|B_{ij}|$ otherwise. It is crucial to note that for $B\geq 0$, $M(B)$ need not be positive semi-definite. In fact, the constraint $M(B)\geq 0$ is far from trivial and is actually used to define the so called class of H-matrices \cite[Chapter 6]{axelsson1996iterative}. However, it turns out that the $\Delta$-freeness of $\rho_{A,B,C}$ is strong enough to ensure that above constraint holds for $B, C\geq 0$, as we now show.
\begin{theorem} \label{theorem:Tfree}
If $\rho_{A,B,C}$ is $\Delta$-free and separable, then $M(B)$ and $M(C)$ are positive semi-definite.
\end{theorem}
\begin{proof}
Since $\rho_{A,B,C}$ is separable, $(A,B,C)$ is TCP and there exist vectors $\{\ket{v_k}, \ket{w_k}\}_{k\in I}\subset \C{d}$ such that the decomposition in Definition~\ref{def:PCP-TCP} holds. Now, if there exists a $k\in I$ such that $\sigma(v_k\odot w_k)\geq 3$, then $G(\ketbra{v_k\odot \overbar{v_k}}{w_k\odot \overbar{w_k}})$ (and hence $G(A)$) would contain a triangle, which is not possible since $\rho_{A,B,C}$ is $\Delta$-free. Hence, $\sigma(v_k\odot w_k)=\sigma(v_k\odot \overbar{w_k})\leq 2$ for each $k$. Now, let
\begin{align*}
    I_1 &= \{k\in I : \sigma(v_k\odot w_k) \leq 1 \} \\
    I_{ij} &=\{k\in I : \operatorname{supp}\ket{v_k\odot w_k} = \{i,j\} \} \quad ( \text{for }i<j)
\end{align*}
so that the index set splits as $I = I_1 \cup_{i<j} I_{ij}$. Further, if we define $B^{ij}=\sum_{k\in I_{ij}} \ketbra{v_k\odot w_k}{v_k\odot w_k}$, it is easy to see that $M(B^{ij}) \geq 0$ for all $i<j$, since each $B^{ij}$ is supported on a $2$-dimensional subspace $\mathbb{C}\ket{i}\oplus \mathbb{C}\ket{j}\subset \C{d}$. The following decomposition then shows that $M(B)\geq 0$:
\begin{align*}
    M(B) = \sum_{k\in I_1}& \ketbra{v_k\odot w_k}{v_k \odot w_k} + \sum_{1\leq i<j\leq d} M (B^{ij})
\end{align*}
A similar argument shows that $M(C)\geq 0$ as well.
\end{proof}

Observe how the $\Delta$-freeness of $G(A)$ above is used to force the vectors $v_k,w_k$ to have small common supports. If we define $A_k=\ketbra{v_k\odot \overbar{v_k}}{w_k\odot \overbar{w_k}}$ and $B_k, C_k$ as the rank-one projections onto $\ket{v_k\odot w_k}, \ket{v_k\odot \overbar{w_k}}$, it is clear that $\rho_{A,B,C}= \sum_{k}\rho_{A_k,B_k,C_k}$, where the small common supports imply that each $\rho_{A_k,B_k,C_k}$ has support on a $2\otimes 2$ subsystem (barring some diagonal entries). This is what ensures that $M(B),M(C)\geq 0$, as was shown in the proof above. Remarkably, for CLDUI states, the converse also holds (i.e. if $\rho_{A,B}$ is PPT and $M(B)\geq 0$, then $\rho_{A,B}$ is separable). \cite[Corollary 5.5]{Johnston2019pairwise}. Hence, we obtain a complete characterization of separable $\Delta$-free CLDUI states. 

\begin{theorem}\label{theorem:Tfree-CLDUI}
If $\rho_{A,B}$ is $\Delta$-free and PPT, then $\rho_{A,B} \text{ is separable} \iff M(B)\geq 0$.
\end{theorem}

For LDOI states, the converse of Theorem~\ref{theorem:Tfree} ceases to hold, see \cite[Example 9.2]{Singh2020diagonal}. Nevertheless, we do have a non-trivial necessary condition for separability, which we now exploit to arrive at our primary entanglement detection strategy in arbitrary $\Delta$-free states.
\begin{theorem}\label{theorem:main_test}
Let $\rho\in \M{d}\otimes \M{d}$ be an arbitrary $\Delta$-free state with the associated matrices $A,B, C\in \M{d}\,$ defined entrywise as $A_{ij}=\langle ij|\rho|ij\rangle$, $B_{ij}=\langle ii|\rho|jj\rangle$, and $C_{ij}=\langle ij|\rho|ji\rangle$ for $i,j\in [d]$. Then, $\rho$ is entangled if either $M(B)$ or $M(C)$ is not positive semi-definite.
\end{theorem}
\begin{proof}
For LDOI states, the conclusion follows from Theorem~\ref{theorem:Tfree}. For an arbitrary $\rho\in \M{d}\otimes \M{d}$, the result follows from the fact that it can be transformed into a $\Delta$-free LDOI state through a separability preserving local operation: $\rho_{A,B,C} = \mathbb{E}_O [(O\otimes O) \rho (O\otimes O)]$.
\end{proof}

We should emphasise here that the above test is very easily implementable. Given a $\Delta$-free state $\rho\in \M{d}\otimes \M{d}$, one simply needs to extract the appropriate entries from $\rho$ to form the $B,C$ matrices and check for the positivity of the corresponding comparison matrices. The test can also be generalized to work for states $\rho\in \M{d_1}\otimes \M{d_2}$ with $d_1< d_2$, as we now illustrate. First define a $d_1\times d_2$ matrix (as before) $A_{ij}=\langle ij|\rho|ij\rangle$, with $i\in [d_1]$, $j\in [d_2]$. Now, corresponding to all possible ways of choosing $d_1$ different columns of $A$, construct a total of $d_2 \choose d_1$ matrices $\widetilde{A}\in \M{d_1}$. For each $\widetilde{A}$ with a $\Delta$-free $G(\widetilde{A})$, construct the projector $P=\sum_{i}\ketbra{i}{i}$, where the sum is over those $i\in [d_2]$ which were chosen to form columns of $\widetilde{A}$. Project the original state $\rho$ locally to obtain a $d_1\otimes d_1$ $\Delta$-free state $\widetilde{\rho} = (I\otimes P)\rho(I\otimes P)$, and apply the test from Theorem~\ref{theorem:main_test}.

\begin{remark}
At this juncture, it is crucial to point out that before applying Theorem~\ref{theorem:main_test}, it is essential to find a suitable product basis in which the given state $\rho$ admits a $\Delta$-free matrix representation (as in Definition~\ref{def:triangle-free}), which seems like a daunting task. We pose this as an open problem in the concluding section.
\end{remark}

With our primary entanglement test in place, we can now begin to analyse the suitability of various kinds of $\Delta$-free graphs from the perspective of Theorem~\ref{theorem:main_test}. More specifically, we wish to describe the class of $\Delta$-free graphs which allows for matrix realizations $B\geq 0$ to exist such that $M(B)\ngeq 0$. In what follows, we will prove that this class contains precisely those $\Delta$-free graphs which contain a cycle (of length $\geq 4$). The first step towards obtaining this characterization is to show that acyclic graphs do not allow for the proposed matrix realizations to exist, which forms the content of our next result.

\begin{proposition}\label{prop:acyclic}
If $B\geq 0$ is a matrix realization of an acyclic graph, then $M(B)\geq 0$.
\end{proposition}
\begin{proof}
Observe that $M(B)\geq 0$ if and only if $\operatorname{det}M(\widetilde{B})$ is non-negative for all principal submatrices $\widetilde{B}$ of $B$ (which also correspond to acyclic graphs $G(\widetilde{B})$ if $G(B)$ is acyclic). Hence, to prove the result, it suffices to show that $\operatorname{det}M(B)$ is non-negative for an arbitrary $B\geq 0$ such that $G(B)=(V,E)$ is acyclic. With this end in sight, we first recall the formula: 
\begin{equation*}
    \operatorname{det}B = \sum_{\sigma\in S_d} B_{\sigma} = \sum_{\sigma\in S_d} (-1)^{\sigma} \prod_{i\in [d]}B(i,\sigma_i)
\end{equation*}
where the summation occurs over all permutations $\sigma$ with $-1^\sigma$ denoting their signs. Now, for some $B_\sigma \neq 0$, we claim that if $\{i,j\}\in E$ is an edge for $i\neq j=\sigma_i$, then $\sigma_j=i$, which implies that $B_\sigma$ contains a factor of $|B_{ij}|^2$. Since this factor stays invariant as $B$ becomes $M(B)$, the preceeding claim implies that $\operatorname{det}M(B)=\operatorname{det}B\geq 0$, which is precisely our requirement. We now prove the claim. Assume on the contrary that $\sigma_j =k\notin \{i,j \}$. Then, since $B_\sigma\neq 0$, $B(j,k)\neq 0\implies \{j,k\}\in E$. Now, if $\sigma_k=i$, then $\{i,j\}, \{j,k\}, \{k,i\}$ forms a triangle, which contradicts the fact that $G(B)$ is acyclic. Hence $\sigma_k=l\notin \{i,j,k\}$ and, as before, $\{k,l \}\in E$. Continuing in the fashion, it becomes evident that the sequence of edges $\{i,j\},\{j,k\},\{k,l\},\ldots$ must eventually terminate in a cycle, which is not possible since $G$ is acyclic.
\end{proof}
\begin{figure}[H]
    \centering
    \includegraphics[scale=1.5]{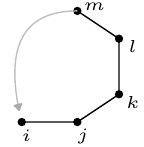}
    \caption{The sequence of edges $\{\{i,j\},\{j,k\},\{k,l\},\ldots \}\subset E$ in the proof of Proposition~\ref{prop:acyclic} eventually terminates in a cycle.}
    \label{fig:my_label}
\end{figure}

If $\rho_{A,B,C}$ is a PPT LDOI state with acyclic $G(A)$ (which implies that $G(B),G(C)$ are acyclic as well), the above proposition implies that $M(B),M(C)\geq 0$, thus rendering Theorem~\ref{theorem:main_test} incapable of detecting any entanglement. For CLDUI states, Theorem~\ref{theorem:Tfree-CLDUI} allows us to trivially prove the following result.
\begin{corollary}
If $\rho_{A,B}$ is such that $G(A)$ is acyclic, then $\rho_{A,B} \text{ is separable} \iff \rho_{A,B} \text{ is PPT}$.
\end{corollary}

With acyclic graphs now out of the picture, the next proposition investigates the remaining class of cyclic $\Delta$-free graphs. This result will form the basis of our method to construct exotic families of PPT entangled $\Delta$-free states in the next section.

\begin{proposition}\label{prop:cyclic}
For a $\Delta$-free cyclic graph, a matrix realization $B\geq 0$ exists such that $M(B)\ngeq 0$.
\end{proposition}
\begin{proof}
We first show that the result holds for $k$-cycles. Assume $k=4$. Let $B=XX^T$, where 
\begin{equation} \label{eq:cyclic-B}
X=\left(
\begin{array}{ c r c}
   1 & 0 & \phantom{-}0 \\
   1 & 1 & \phantom{-}0 \\
   0 & 1 & \phantom{-}1 \\
   1 & -1 & \phantom{-}0
  \end{array}
\right) \quad B =\left(
\begin{array}{ c c r r}
   1 & \phantom{-}1 & 0 & 1 \\
   1 & \phantom{-}2 & 1 & 0 \\
   0 & \phantom{-}1 & 2 & -1 \\
   1 & \phantom{-}0 & -1 & 2
  \end{array}
\right)    
\end{equation}
Clearly, $B\geq 0$, $G(B)=C_4$ and $M(B)\ngeq 0$. Similar construction can be employed for arbitrary $k>4$, by defining $X$ as a $k\times (k-1)$ matrix entrywise: $X_{ii}=1$ for $i\in [k-1]$, $X_{i+1,i}=1$ and $X_{ki}=(-1)^{i+1}$ for $i\in [k-2]$, and $B=XX^T$. It is then easy to see see that $B\geq 0$, $G(B)=C_4$ and $M(B)\ngeq 0$. Now, let $G$ be an arbitrary $\Delta$-free graph on $d$ vertices with $C_k\subseteq G$. Let $B'\in \M{d}$ be such that it contains the above constructed $B$ as the principal submatrix corresponding to the vertices which form the cycle $C_k$, with all other entries defined to be zero. Let the adjacency matrix of $G$ be $\operatorname{ad}G$. Then, it is easy to see that for every $x>0, \, \exists \, 0\neq y\in \C{}$ with $|y|\leq x$ such that $xI_d + y\operatorname{ad}G\geq 0$. Now, observe that $B_x = B' + xI_d + y\operatorname{ad}G\geq 0$ for all $x> 0$, and $B' = \operatorname{lim}_{x\rightarrow 0^+}B_x$ is such that $M(B')\ngeq 0$. Since the cone of positive semi-definite matrices is closed in $\M{d}$, we can deduce that there exists an $x>0$ such that $B=B_x\geq 0$, $M(B)\ngeq 0$ and $G(B)=G$, thus finishing the proof.
\end{proof}

\section{Construction of new families of PPT entangled $\Delta$-free states} \label{sec:PPTent}
By exploiting the property of cyclic $\Delta$-free graphs from Proposition~\ref{prop:cyclic}, we now present a simple protocol for constructing new families of PPT entangled $\Delta$-free states in arbitrary $d_1\otimes d_2$ dimensions ($d_1,d_2 \geq 4$): \\[0.15cm]
\emph{Step 1.} Choose a $\Delta$-free cyclic graph $G$ on $d$ vertices. \\[0.2cm]
\emph{Step 2.} Construct a matrix realization $B\in \M{d}$ of $G$ such that $B\geq 0$ and $M(B)\ngeq 0$. \\[0.2cm]
\emph{Step 3.} Construct a family $F_B$ of matrix pairs $(A,C)$ in the following manner: choose $A\succcurlyeq 0, \,C\geq 0$ such that $G(A)=G(B)=G(C)=G, \,\,\operatorname{diag}A=\operatorname{diag}B=\operatorname{diag}C$ and $A_{ij}A_{ji}\geq \operatorname{max}\{|B_{ij}|^2,|C_{ij}|^2 \}\,\, \forall i,j$. (Impose $\sum_{i,j}A_{ij}=1$ to ensure trace normalization). Then, the two matrix-parameter ($\sim d^2$ real parameters) family $\{\rho_{A,B,C} \}_{(A,C)\in F_B}$ represents a class of $d\otimes d$ $G$-PPT entangled LDOI states, associated with the $\Delta$-free cyclic graph $G$ and its matrix realization $B\geq 0$ such that $M(B)\ngeq 0$. The construction can be easily generalized to arbitrary dimensions with $d_1\neq d_2$ with the help of the discussion following Theorem~\ref{theorem:main_test} in Section~\ref{sec:main}.\\[0.1cm]

Let us use the above method to explicitly construct a $4\otimes 4$ PPT entangled family of $\Delta$-free LDOI states. 
\begin{example}
We begin by choosing the unique $\Delta$-free (connected) graph on $4$ vertices: the 4-cycle $C_4$. Next, we take the matrix realization $B\geq 0$ ($M(B)\ngeq 0$) of $C_4$ from Eq.~\eqref{eq:cyclic-B}. To construct the family $F_B$ of matrix pairs $(A,C)$, we proceed as follows. First observe from the general form of $A,C$ given in Eq.~\eqref{eq:A,B,C} that $A,B,C$ have equal diagonals and $G(A)=G(B)=G(C)=C_4$. It should be noted that even though $a_{13}$ may be non-zero, $G(A)$ doesn't contain the edge $\{1,3\}$ as $a_{31}=0$.  Moreover, it is easy to choose complex numbers $c_{ij}$ such that $C\geq 0$. For this illustration, we simply impose the constraint that $c_{ii}\geq \sum_{j\neq i} |c_{ij}|$ for each $i\in [4]$, so that $C$ becomes (hermitian) diagonally dominant and hence positive semi-definite. Finally, we let $a_{ij}$ be non-negative real numbers such that $a_{ij}a_{ji}\geq \operatorname{max}\{|b_{ij}|^2,|c_{ij}|^2\} \, \forall i,j\in [4]$. With these constraints in place, we obtain our family of (unnormalized) $C_4$-PPT entangled $\Delta$-free LDOI states with $\sim 18$ real parameters: $ \{\rho_{A,B,C}\in \M{4}\otimes \M{4} : (A,C)\in F_B \} $. Trace normalization can be enforced by using the fact that $\operatorname{Tr}\rho_{A,B,C} = \sum_{i,j}A_{ij}$.
\begin{equation} \label{eq:A,B,C}
G(A) = \includegraphics[scale=1.4,align=c]{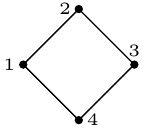} \quad A=\left(
\begin{array}{ c c c c}
   1 & a_{12} & a_{13} & a_{14} \\
   a_{21} & 2 & a_{23} & 0 \\
   0 & a_{32} & 2 & a_{34} \\
   a_{41} & a_{42} & a_{43} & 2
  \end{array}
\right)  
\end{equation}
\begin{equation*} 
\quad B =\left(
\begin{array}{ c c r r}
   1 & \phantom{-}1 & 0 & 1 \\
   1 & \phantom{-}2 & 1 & 0 \\
   0 & \phantom{-}1 & 2 & -1 \\
   1 & \phantom{-}0 & -1 & 2
  \end{array}
\right) \quad C=\left(
\begin{array}{ c c c c}
   1 & c_{12} & 0 & c_{14} \\
   \overbar{c_{12}} & 2 & c_{23} & 0 \\
   0 & \overbar{c_{23}} & 2 & c_{34} \\
   \overbar{c_{14}} & 0 & \overbar{c_{34}} & 2
  \end{array}
\right)   
\end{equation*}
\end{example}

It seems wise to pause here for a moment to appreciate the richness of the above construction method. There is an immense variety of $\Delta$-free (cyclic) graphs, which directly translates into a similar variety in the realm of PPT entangled $\Delta$-free states. To gain perspective, we list the number of distinct (connected) $\Delta$-free cyclic graphs on $d\geq 4$ unlabelled vertices in the following sequence (obtained by subtracting the sequences \cite[\href{https://oeis.org/A024607}{A024607} --      \href{http://oeis.org/A000055}{A000055}]{oeis}):
\begin{equation}
    1, 3, 13, 48, 244, 1333, 9726, 90607, 1143510 \ldots
\end{equation}
For instance, in a $15\otimes 15$ system, the $\sim 10^{10}$ (connected) $\Delta$-free cyclic graphs correspond to $\sim 10^{10}$ distinct classes of PPT entangled $\Delta$-free states. Within each class, the $\sim 15^2$ real parameters and different matrix realizations $B$ of the respective graphs only furthers the diversity. \\[0.1cm]

\section{Going beyond $\Delta$-free states} \label{sec:beyond}
Let $\rho\in \M{d}\otimes \M{d}$ be such that the associated matrix $A$ has a triangle containing graph $G(A)$, which means that the test in Theorem~\ref{theorem:main_test} is inapplicable. Nevertheless, there may exist induced subgraphs $\widetilde{G}$ within $G(A)$ [these are of the form $G(\widetilde{A})$ for $\widetilde{d}\times \widetilde{d}$ principal submatrices $\widetilde{A}$ of $A$, $\widetilde{d}<d$] which are $\Delta$-free. For such a $\widetilde{G}$, let us define the projector $P=\sum_{i}\ketbra{i}{i}$, where the sum runs over those rows/columns $i\in [d]$ which are present in $\widetilde{A}$. Clearly, $(P\otimes P)\rho(P\otimes P)$ is then $\Delta$-free and hence a valid candidate for Theorem~\ref{theorem:main_test}. We succinctly describe the above discussion in the form of a theorem below.

\begin{theorem} \label{theorem: nonTfree}
Consider an arbitrary $\rho\in \M{d}\otimes \M{d}$ with the associated matrices $A,B, C\in \M{d}\,$ defined as $A_{ij}=\langle ij|\rho|ij\rangle$, $B_{ij}=\langle ii|\rho|jj\rangle$, and $C_{ij}=\langle ij|\rho|ji\rangle$ for $i,j\in [d]$. If there exists a principal submatrix $\widetilde{A}$ of $A$ such that $G(\widetilde{A})$ is $\Delta$-free and either $M(\widetilde{B})$ or $M(\widetilde{C})$ is not positive semi-definite, then $\rho$ is entangled.
\end{theorem}
Using Theorem~\ref{theorem: nonTfree}, we have been able to successfully detect entanglement in several randomly generated $d\otimes d$ non $\Delta$-free LDOI states (where $d\sim20$). Let us see a simple exhibit of how this method works.
\begin{example}
Consider an unnormalized $6\otimes 6$ PPT CLDUI state $\rho_{A,B}$ with matrices $A,B\in \M{6}$ defined in Eq.~\eqref{eq:V6-1}. Clearly, Theorem~\ref{theorem:main_test} is not applicable here, since $G(A)$ contains a lot of triangles. However, just by removing the $5$\textsuperscript{th} row and column from $A$, we obtain the principal submatrix $\widetilde{A}$ with a nice $\Delta$-free graph $G(\widetilde{A})$, see Eq.~\eqref{eq:V6-2}. Moreover, since $M(\widetilde{B})\ngeq 0$, Theorem~\ref{theorem: nonTfree} tells us that $\rho_{A,B}$ is entangled.
\begin{widetext}
\begin{alignat}{3}
G(A) &= \includegraphics[scale=1, align=c]{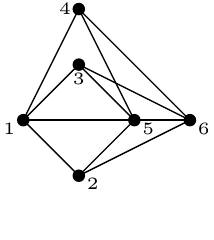} \qquad A &&=\left(
\begin{array}{ c c c c c c}
   11 & 9 & 6 & 6 & 4 & 0 \\
   10 & 13 & 4 & 1 & 11 & 8 \\
   5 & 0 & 13 & 4 & 7 & 8 \\
   6 & 0 & 0 & 13 & 11 & 12 \\
   2 & 9 & 5 & 14 & 15 & 14 \\
   4 & 10 & 4 & 10 & 14 & 11 
  \end{array}
\right) \quad B &&=\left(
\begin{array}{ r r r r r r}
   11 & -7 & 1 & -3 & -1 & 0 \\
   -7 & 13 & 0 & 0 & 6 & 7 \\
   1 & 0 & 13 & 0 & -2 & 3 \\
   -3 & 0 & 0 & 13 & -9 & -8 \\
   -1 & 6 & -2 & -9 & 15 & 10 \\
   0 & 7 & 3 & -8 & 10 & 11
  \end{array}
\right)  \label{eq:V6-1}  \\
G(\widetilde{A}) &= \includegraphics[scale=1, align=c]{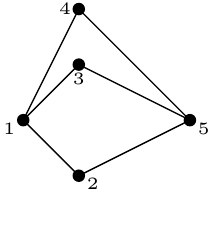} \qquad \widetilde{A} &&=\left(
\begin{array}{ c c c c c c}
   11 & 9 & 6 & 6 & 0 \\
   10 & 13 & 4 & 1 & 8 \\
   5 & 0 & 13 & 4 & 8 \\
   6 & 0 & 0 & 13 & 12 \\
   4 & 10 & 4 & 10 & 11 
  \end{array}
\right) \quad\quad \,\,\widetilde{B} &&=\left(
\begin{array}{ r r r r r r}
   11 & -7 & 1 & -3 & 0 \\
   -7 & 13 & 0 & 0 & 7 \\
   1 & 0 & 13 & 0 & 3 \\
   -3 & 0 & 0 & 13 & -8 \\
   0 & 7 & 3 & -8 & 11
  \end{array}
\right) \label{eq:V6-2}
\end{alignat}
\end{widetext}
\end{example}

It is relevant to point out that the complexity of the above method scales badly with $d$, since the problem of determining the existence of $\Delta$-free induced subgraph in a graph is NP-complete. Let us quickly prove this. It is trivial to check that the problem is NP. Now, for a graph $G=(V,E)$, it is not too hard to see that $G$ contains an independent set of size $k$ (this is an induced subgraph on $k$ vertices with no edges) if and only if $H_G$ contains a $\Delta$-free induced subgraph with $|E|+k$ vertices, where $H_G$ is constructed from $G$ by adding vertices $v_e$ and edges $\{i,v_e\}$ and $\{j,v_e\}$ for each edge $e=\{i,j\}\in E$. Hence, the NP-completeness of the problem of determining whether a graph contains an independent set or not \cite{Garey1979computers} imparts a similar hardness to the problem of finding $\Delta$-free induced subgraphs within a given graph.

\section{Concluding discussion} \label{sec:conclusion}
In this paper, we have presented a unique test to detect a new kind of bipartite entanglement which is present in states with peculiar $\Delta$-free distribution of zeros on their diagonals. Such a connection between the entanglement of a state and its diagonal zero pattern is previously unheard of. From our recipe to construct families of PPT entangled $\Delta$-free states in arbitrary dimensions, it is evident that the ease of $ \Delta $-free entanglement detection does in no way restrict its diversity. We have also established an intriguing link between the problems of detecting entanglement in non $\Delta$-free states and finding $\Delta$-free induced subgraphs within a given graph. Several avenues of research stem from our work. The most obvious question to ask is whether the usual entanglement criteria -- such as the realignment \cite{Rudolph2000realignment, Kai2002realignment} or the covariance \cite{Eisert2008cov} criterion -- can detect entanglement in $\Delta$-free states? In a more practical setting (especially when the full state-tomography is impossible \cite{Huber2014ent}), one would like to know the structure of the entanglement witnesses which can detect this kind of entanglement. These questions are not straightforward to answer because of the incredibly diverse nature of $\Delta$-free entanglement. For example, even if we consider the most trivial family of $C_4$-PPT entangled $\Delta$-free states from Section~\ref{sec:PPTent}, the $18$ real parameters inside provide the states with ample freedom to evade detection from any of the usual entanglement tests. We have even been able to tweak the parameters so that the semi-definite hierarchies from \cite{Doherty2004alg, Eisert2004alg} give up on detecting entanglement in any reasonable time-frame. Hence, one can deduce that our comparison matrix entanglement test is highly non-trivial and has strong potential to provide drastic computational speed-ups over its regular counterparts in a variety of circumstances. We now conclude our discussion with a few pertinent remarks and open problems.

Firstly, as has already been pointed out, the property of $\Delta$-freeness of a state is basis-dependent. In other words, if $\rho\in \M{d}\otimes \M{d}$ is $\Delta$-free and $U,V\in \M{d}$ are unitary matrices, then $(U\otimes V)\rho (U\otimes V)^{\dagger}$ need not be $\Delta$-free. Thus, a natural question arises: Is there a basis-independent description of the $ \Delta $-free property of a state? More specifically, given $\rho\in \M{d}\otimes \M{d}$, how does one guarantee the existence of local unitaries $U,V\in \M{d}$ such that $(U\otimes V)\rho (U\otimes V)^{\dagger}$ is $\Delta$-free? The answer to the above question can give us insights into what it physically means for a state to be $\Delta$-free and hence provide us with a deeper understanding of the nature of $\Delta$-free entanglement itself. Other entanglement-theoretic properties of $\Delta$-free states (distillability, entanglement cost, etc.) deserve further scrutiny.

Secondly, observe that our main result relies heavily on Theorem~\ref{theorem:Tfree}, where the idea is to show that the vectors $\{\ket{v_k},\ket{w_k} \}_{k\in I}$ in the TCP decomposition of $(A,B,C)$ have small common support ($\sigma(v_k\odot w_k)\leq 2$ for each $k$). While the $\Delta$-freeness of $\rho_{A,B,C}$ is sufficient to guarantee this, it is clearly not necessary. Thus, other constraints on $\rho_{A,B,C}$ which ensure that the above property holds can allow one to detect analogues of $\Delta$-free entanglement. More generally, a hierarchy of constraints $\sigma(v_k\odot w_k)\leq n$ for $n\in \mathbb{N}$ on vectors $\{\ket{v_k},\ket{w_k} \}_{k\in I}$ in the TCP decompositions of $(A,B,C)$ can be analysed to see if they entail simple necessary conditions on separability of $\rho_{A,B,C}$. Significant progress along these directions has been made in \cite[Section 3]{singh2020ppt2}, albeit in a different context.  
\\
%TC:ignore
\emph{Acknowledgements.} I would like to thank Ion Nechita for helpful insights and stimulating discussions on this topic. This work is partially supported by an INSPIRE Scholarship for Higher Education by the Department of Science and Technology, Government of India. 
\PRLsep

\bibliography{references}

\end{document}